\documentclass[final,5p,times,twocolumn]{elsarticle}
\usepackage{adjustbox}
\usepackage{algpseudocode}
\usepackage{algorithm}
\usepackage{amssymb}
\usepackage{amsmath}
\usepackage{amsthm}    
\usepackage{array}
\usepackage[english]{babel}
\usepackage{bbold}
\usepackage{float}
\usepackage{graphicx} 
\usepackage{hyperref}
\usepackage[utf8]{inputenc}
\usepackage{lipsum}
\usepackage[cal=boondoxo]{mathalpha}
\usepackage{mathtools}
\usepackage{moresize}
\usepackage{multicol}
\usepackage{notoccite}
\usepackage{todonotes}
\usepackage{stmaryrd}
\usepackage{siunitx}
\usepackage{subfig}
\usepackage{textcomp}
\usepackage{tikz}
\usepackage{url}
\usepackage{multirow}
\usepackage{adjustbox}
\usepackage{rotating}
\usetikzlibrary{positioning}
\usepackage{csquotes}
\usepackage{natbib}

\newtheorem{proposition}{Proposition}

\begin{document}
\begin{frontmatter}
\title{Pricing Filtering in Dantzig-Wolfe Decomposition}
\author[1]{Abdellah Bulaich Mehamdi \corref{myfootnote}}
\cortext[myfootnote]{Corresponding author: Abdellah Bulaich Mehamdi. \\
Palaiseau, France, 91120.\\
E-mail address: \href{mailto:abdellah.bulaich-mehamdi@polytechnique.edu}{abdellah.bulaich-mehamdi@polytechnique.edu}. \\
Currently at EDF Lab Paris-Saclay.}
\author[2]{Mathieu Lacroix}
\author[1]{Sébastien Martin}

\affiliation[1]{organization={Huawei Technologies Ltd., Paris Research Center},
city={Boulogne-Billancourt},
country={France}}

\affiliation[2]{organization={Laboratoire d'Informatique de Paris-Nord (LIPN), Université Sorbonne Paris Nord - CNRS},
city={Villetaneuse},
country={France}}

\begin{abstract}
Column generation is used alongside Dantzig-Wolfe Decomposition, especially for linear programs having a decomposable pricing step requiring to solve numerous independent pricing subproblems.
We propose a filtering method to detect 
which pricing subproblems may have improving columns, and only those subproblems are solved during pricing. 
This filtering is done by providing light, computable bounds using dual information from previous iterations of the column generation. The experiments show a significant impact on different combinatorial optimization problems.
\end{abstract}

\end{frontmatter}

\section{Introduction}
\emph{Dantzig-Wolfe decomposition (DWD)} \cite{DWDpaper} is a well-known algorithm to solve linear optmization problems of the following form:
\begin{align*}
    &\min &&\sum_{k \in K} c^k \cdot x^k  \\
    &\text{s.t. } &&\sum_{k \in K} A^k \, x^k \geq b \\
    & &&x^k \in X^k &\forall k \in K      \\
    & &&x \geq 0 &\forall k \in K 
    \end{align*}
where $X^k$ is a polytope for each $k \in K$. This method consists in replacing each variable $x^k$ by a convex combination of the extreme points $\mathcal{X}^k = \{\chi_p^k: p \in P^k\}$ of $X^k$ for each $k \in K$, where a variable $\lambda_p^k$ is associated with each $\chi_p^k$, $p \in P^k$, leading to the following equivalent linear optimization problem. 
\begin{align}
&\min &&\sum_{k \in K} \sum_{p \in P^k} (c^k \cdot \chi_p^k ) \,  \lambda_p^k \label{dwd:obj} \\
&\text{s.t. } &&\sum_{k \in K} \sum_{p \in P^k} (A^k \, \chi_p^k) \, \lambda_p^k \geq b     \label{dwd:linking} \\
& &&\sum_{p \in P^k} \lambda_p^k = 1 &\forall k \in K     \label{dwd:convexity} \\
& &&\lambda_p^k \geq 0 &\forall k \in K, \forall p \in P^k \label{dwd:nonnegativity}
\end{align}

Since Formulation \eqref{dwd:obj}-\eqref{dwd:nonnegativity} usually contains an exponential number of variables, DWD consists in solving it with a \emph{delayed column generation (CG)}. A \emph{restricted master problem (RMP)}, corresponding to 
\eqref{dwd:obj}-\eqref{dwd:nonnegativity} but restricted to a subset of variables, is solved. Then, a \emph{pricing step} occurs which can be decomposed into $|K|$ pricing subproblems. The \emph{pricing subproblem} associated with $k \in K$ corresponds to finding a column of $\mathcal{X}^k$ with a negative reduced cost or to state that such a column does not exist in $\mathcal{X}^k$. If no subproblem returns a negative reduced cost column, the solution of RMP is optimal for \eqref{dwd:obj}-\eqref{dwd:nonnegativity}. Otherwise, all columns with a negative reduced cost returned by the subproblems are added to RMP and the algorithm iterates.

\medskip

In DWD, usually a lot of pricing subproblems are solved but return no negative reduced cost column. In this paper, we present a \emph{pricing subproblem filtering step} that is performed in each iteration of DWD to determine which pricing subproblems to solve in this iteration. This filtering step is done by computing for each $k \in K$ a lower bound on the minimum reduced cost among those of $\mathcal{X}^k$, and solving only the associated pricing subproblem when this lower bound is negative, as otherwise, there is no negative reduced cost column in $\mathcal{X}^k$. By decreasing the overall number of pricing subproblems solved during DWD, it improves its performance.

\medskip

In Section~\ref{sec:literature}, we present a literature overview related to our work. In Section~\ref{sec:filtering}, we present how to compute lower bounds on the minimum reduced cost. We also explain how to use it to filter pricing subproblems. We also propose a heuristic method to improve this lower bound. This allows to filter more pricing subproblems but sometimes filters ones for which there may exist a negative reduced cost column. 
In Section~\ref{sec:evaluation}, we evaluate the impact of different filtering strategies in DWD when solving the linear relaxation of two problems: the generalized assignment problem and the multi-commodity flow problem with side constraints.

\section{Related work}\label{sec:literature}
 
We are not aware of any other method that filters pricing subproblems in DWD based on a lower bound of the best-reduced cost of a pricing subproblem. However, reducing the running time of DWD or CG is an active area of research for decades.

To reduce the time at each iteration of DWD, one can stop the resolution of the pricing subproblems when enough negative reduced cost columns are found. For instance, in~\cite{Gamache}, the Aircrew rostering problem is solved using a solution method based on CG. The pricing problem consists in solving for each employee a shortest path problem with several resource constraints. Since up to 95\% of the CG time is spent in solving the pricing subproblems, they stop the pricing at each iteration of CG when a fixed number of negative reduced cost columns has been found. 

Another way to speed up the pricing is to run heuristics to find negative reduced cost columns and solve the pricing in an exact way only if the heuristics fail to find ones, see for instance~\cite{mehrotra_column_1996, shen_enhancing_2022} for greedy heuristics for the pricing problem when solving the graph coloring problem using a branch-and-price algorithm.

It is also possible to relax some constraints of the pricing problem to simplify its resolution. However, this generates non feasible columns which may give infeasible solutions and degrade the quality of the DWD bound. This has been done successfully for different vehicle routing problems. In such problems, the pricing reduces to computing resource-constrained elementary shortest path problems which are NP-hard. Removing (at least partially) the elementary requirement accelerates the pricing resolution but degrades the quality of the linear relaxation~\cite{Irnich2006TheSP}. This degradation is usually limited by adding strengthening inequalities~\cite{fukasawa_robust_2006}. 

Finally, CG generally suffers from the dual variables' oscillations between successive iterations due to the degeneracy of the linear problem. Numerous stabilization methods have been proposed over the years to tackle and address this issue: box step method~\cite{marsten_b_1975}, linear penalty function~\cite{DUMERLE1999229}, bundle method \cite{lemarechalNonsmooth1978,kiwielAggregate1983}, smoothing and in-out separation~\cite{pessoa_automation_2018} and dual optimal inequalities \cite{ben_amor_dual-optimal_2006,valerio_de_carvalho_using_2005,gschwind_dual_2016} 
to name a few. These stabilization methods tend to decrease the number of iterations of CG or DWD, improving the overall running time.

\medskip

Note that our filtering approach can be used in combination with the different mentioned methods. For instance, one can stop the pricing phase when enough negative reduced cost columns are found, the search of these columns being done by solving pricing subproblems that are not filtered with our method. Moreover, for each of these subproblems, one can run heuristics to try to fastly find negative reduced cost columns. As our method is generic, it is possible to relax some constraints in the pricing subproblems to speed up their resolution. Finally, since our filter method uses dual variables, it can also be applied when a stabilization method is used on the restricted master problem.

\section{Pricing filtering}\label{sec:filtering}

\subsection{Exact filtering}

When it comes to subproblems, a basic observation is that solving a subproblem that has a nonnegative minimum reduced cost is irrelevant. Thus, having an idea about its minimum reduced cost sign before computing it can help us to decide whether or not to solve it. In this context, we investigate the column generation's iterative structure and seek to establish a computable light bound to predict the minimum reduced cost sign in advance. By doing so, we hope to reduce the number of solved subproblems without compromising the solution quality and without heavy computations.

\medskip

Let RMP$^t$ denote RMP at iteration $t \geq 1$. RMP$^t$ contains for $k \in K$ a subset of extreme points of $X^k$ given by $\{\chi_p^k: p \in P^{k,t}\}$. Let $(\pi^t, \mu^t)$ be a dual optimal solution of RMP$^t$ where $\pi^t$ and $\mu^t$ are associated with inequalities~\eqref{dwd:linking} and \eqref{dwd:convexity}, respectively. The reduced cost of an extreme point $\chi_p^k$ of $X^k$ at iteration $t$ is $(c^k - \pi^t A^k) \cdot \chi_p^k - \mu^t_k$. Let $\bar{c}^{k,t}$ denote the minimum reduced cost at iteration $t$ among those of $\mathcal{X}^k$. Let $Y^k$ be a polyhedron containing $X^k$. 
 The following proposition gives a lower bound of $\bar{c}^{k,t}$ from $\bar{c}^{k,\ell}$ computed at iteration $1 \leq \ell < t$. 

\begin{proposition}
\label{pricing:exact:prop}
For $k \in K$ and $1 \leq \ell < t$, we have:
\begin{align*}
\bar{c}^{k,t} & \geq \bar{c}^{k,\ell} +  \mu_k^\ell -  \mu_k^t + \min_{x \in Y^k} ( (\pi^\ell - \pi^{t}) A^k) \cdot x 
\end{align*}
\end{proposition}

\begin{proof}
By definition of $\bar{c}^{k,t}$, the property of the $\min$ function and since $X^k$ is contained in $Y^k$, we have:
\begin{align*}
\bar{c}^{k,t} &= \min_{x \in X^k}  (c^k - \pi^{t} A^k) \cdot x - \mu_k^t\\
& = \min_{x \in X^k}  (c^k - \pi^{t} A^k) \cdot x - \mu_k^t + \bar{c}^{k,\ell} - \min_{x \in X^k}  (c^k - \pi^{\ell} A^k) \cdot x + \mu_k^\ell\\
& \geq \bar{c}^{k,\ell} +  \mu_k^\ell -  \mu_k^t + \min_{x \in X^k} ( (\pi^\ell - \pi^{t}) A^k) \cdot x\\
& \geq \bar{c}^{k,\ell} +  \mu_k^\ell -  \mu_k^t + \min_{x \in Y^k} ( (\pi^\ell - \pi^{t}) A^k) \cdot x
\end{align*}
\end{proof}

Proposition~\ref{pricing:exact:prop} provides a lower bound on the best reduced cost of each pricing subproblem using the information obtained when solving it to optimality at a previous iteration.
Indeed, at current iteration $t$ of CG, for each $k \in K$ and each previous iteration $\ell < t$, if the value $\bar{c}^{k,\ell}$ has been computed and dual values have also been stored, $\bar{c}^{k,\ell} +  \mu_k^\ell -  \mu_k^t + \min_{x \in Y^k} ( (\pi^\ell - \pi^{t}) A^k) \cdot x$ is a lower bound on $\bar{c}^{k,t}$. If this lower bound is nonnegative, this implies that $\mathcal{X}^k$ has no negative reduced cost column so it is useless to solve the pricing subproblem associated with $k$.
We propose to enhance DWD by filtering pricing subproblems, that is, not solving pricing subproblems for which we are sure there is no negative reduced cost column thanks to the lower bound on the best reduced cost provided in Proposition~\ref{pricing:exact:prop}. Note that to compute the lower bound on $\bar{c}^{k,t}$ given by Proposition~\ref{pricing:exact:prop}, one needs to optimize a linear function over $Y^k$, and this latter should be chosen so that this optimization is efficient while providing a tight bound. For instance, as in our experiments, $Y^k$ may be an hypercube containing $X^k$.

Algorithm~\ref{alg:dwd_pricing_filtering} presents our modified DWD enhanced with filtering. At each iteration $t$, RMP$^t$ is solved to optimality (line 7), and the dual variable vector $\pi_t$ is stored in $\Pi$ (line 8). In this way, it is possible to retrieve all the previous computed dual values $\pi^{\ell}$ with $\ell < t$. The pricing step is decomposed into $|K|$ subproblems.  For each $k \in K$, a filtering step (lines 10-17) is performed to determine whether it is necessary to solve the pricing subproblem associated with $k$ (case $filter$ equals $False$) or not (case $filter$ equals $True$). This corresponds to computing lower bounds on the minimum reduced cost among those of  $\mathcal{X}^k$ (line 12) using Proposition~\ref{pricing:exact:prop} and previous exact resolutions of this pricing subproblems, and checking whether there exists one which is nonnegative (lines 13-16). In this case, no negative cost column exists in $\mathcal{X}^k$ and the pricing subproblem associated with $k$ is not solved. Otherwise, this latter is solved in an exact way (line 19) and returns a minimum reduced cost column $s^{k,t}$ with reduced cost $\bar{c}^{k,t}$. The triplet 
$(t,\bar{c}^{k,t}, \mu_t^k)$ is stored in $\Upsilon_k$ (line 20) in order to be used in the next iterations to compute lower bounds on the minimum reduced cost of $\mathcal{X}^k$. If the computed reduced cost is negative, the associated column is added to the restricted master problem (lines 21-24). DWD ends when no negative reduced cost columns are found during an iteration (case $optimal$ equals $True$).

\begin{algorithm}[t]
    \caption{DWD with pricing filtering}\label{alg:dwd_pricing_filtering}
    \begin{algorithmic}[1]
    \Require RMP$^0$
    \Ensure Optimal solution of RMP
    \State $t \gets 0$
    \State $\Pi \gets \emptyset$
    \State $\Upsilon_k \gets \emptyset$ for all $k \in K$
    \Repeat
        \State $optimal \gets True$
        \State $t \gets t + 1$
        \State $\mu^t,\pi^t$ = dual solution(RMP$^t$)
        \State push $\pi^t$ in $\Pi$
        \For{$k \in K$}
            \State $filter \gets False$
            \ForAll{$(\ell,\bar{c}^{k,\ell}, \mu_\ell^k) \in \Upsilon_k$}
                \State $LB \gets \bar{c}^{k,\ell} +  \mu_k^\ell -  \mu_k^t + \min_{x \in Y^k} ( (\pi^\ell - \pi^{t}) A^k) \cdot x$
                \If{$LB \ge 0$}
                \State $filter \gets True$
                \State \textbf{break}
                \EndIf
            \EndFor
            \If{$\neg filter$}
                \State $\bar{c}^{k,t}, s^{k,t} \gets$ pricing subproblem ($k, \mu_k^t, \pi^t$)
                \State push $(t,\bar{c}^{k,t}, \mu_t^k)$ in $\Upsilon_k$
                \If{$\bar{c}^{k,t} < 0$}
                    \State $optimal \gets False$
                    \State Add column $s^{k,t}$ to RMP$^t$
                \EndIf
            \EndIf
        \EndFor
    \Until{$optimal$}

    \end{algorithmic}
    \end{algorithm}
    
Note that if the pricing subproblem associated with some $k \in K$ is solved in a heuristic way, the reduced cost returned by the heuristic should not be stored in $\Upsilon_k$ as it cannot be used to find valid lower bounds in the next iterations, even if it is negative. Indeed, Proposition~\ref{pricing:exact:prop} requires that $\bar{c}^{k,\ell}$ is the minimum reduced cost of the pricing subproblem associated with $k$ at iteration $\ell$.

\medskip

DWD presented in Algorithm~\ref{alg:dwd_pricing_filtering} stores dual vectors $\pi^t$ at each iteration $t$. If the dimension of $\pi^t$ is usually small since many constraints have been shifted to the pricing subproblems, it can be memory consuming when the number of iterations of DWD increases. In this case, one can modify Algorithm~\ref{alg:dwd_pricing_filtering} to set up a dynamic memory management for storing these dual vectors, storing for instance only the last $\alpha$ vectors where $\alpha$ is a fixed value, or removing dual vectors for which only a small number of pricing subproblems have been solved to optimality since in this case, keeping such a dual vector will allow to compute only a small number of lower bounds at each iteration. In the same way, for each $k\in K$, it can be interesting to remove some values in $\Upsilon_k$ to decrease the number of computed lower bounds, especially if optimizing on $Y^k$ is not efficient.

\subsection{Heuristic filtering}\label{pricing:heuristic}

Proposition~\ref{pricing:exact:prop} yields a lower bound used to filter pricing subproblems, see Algorithm~\ref{alg:dwd_pricing_filtering}. To fasten DWD, we propose to compute a variant of the lower bound of Proposition~\ref{pricing:exact:prop} when $Y^k \subseteq \{x : x \geq \mathbf{0}\}$. Note that the new value is no more a valid lower bound of the best reduced cost of a pricing subproblem. Hence, using it for filtering destroys the exactness of DWD, that is, the returned solution is still primal feasible but may not be optimal since some columns may be missing in the last RMP. However, computational experiments show that there is only a small degradation of the quality of the solution, whereas it speeds up a lot the running time.

The lower bound of Proposition~\ref{pricing:exact:prop} depends on the value $\min_{x \in Y^k} ( (\pi^\ell - \pi^{t}) A^k) \cdot x$. If $Y^k$ is not tight with respect to $X^k$, the lower bound is not tight with respect to the best reduced cost among those of the extreme points of $X^k$. To bypass this issue, we consider a new objective function in this minimization problem.

At iteration $t$, for $k \in K$, let $I(k,t)$ be the indices of the constraints~\eqref{dwd:linking} containing at least one variable of $\{\lambda^k_p : p \in P^{k,t}\}$. Consider the vector $w$ defined by $w_i = \max\{\bigl((\pi^\ell - \pi^{t}) A^k\bigr)_i, z_i\}$ where $z_i$ equals $-\infty$ if $i \in I(k,t)$, and 0 otherwise. Whenever $Y^k \subseteq \{x : x \geq \mathbf{0}\}$, $\min_{x \in Y^k} w \cdot x \ge \min_{x \in Y^k} ( (\pi^\ell - \pi^{t}) A^k) \cdot x$. The heuristic filtering consists in modifying Algorithm~\ref{alg:dwd_pricing_filtering} by replacing line 12 by:
\begin{equation*}
    LB \gets \bar{c}^{k,\ell} +  \mu_k^\ell -  \mu_k^t + \min_{x \in Y^k} w \cdot x
\end{equation*}
Whenever $Y^k \subseteq \{x : x \geq \mathbf{0}\}$, this gives a higher lower bound (that may be not valid), and more pricing subproblems are filtered.

\section{Evaluation}\label{sec:evaluation}

We evaluate our approach on two problems: the multi-commodity flow problem with side constraints and the generalized assignment problem.
Both problems are usually tackled by considering a Dantzig-Wolfe decomposition and are solved using a branch-and-price. We study here the impact of filtering the pricing subproblems when solving the linear relaxation by DWD. 

\medskip

For computing the lower bound for filtering, we consider $Y^k$ as the binary hypercube for both problems. Indeed, preliminary computational results performed on tighter relaxations of $X^k$ for the multi-commodity flow problem with side constraints show slight improvements of the lower bound (with more effort to compute this bound) but not enough to filter more pricing subproblems. In particular, for $k \in K$, we tested when $Y^k$ is defined as the incidence vectors of the arc sets with an arc leaving $s_k$, no arc leaving $t_k$, and at most one arc leaving each other vertex. We also tested when $Y^k$ is the incidence vectors of sets of arcs with at most $\alpha_k$ arcs, where $\alpha_k$ is the maximum number of arcs such that the sum of their delay $d_a$ is no more than $d_k$.

Since the hypercube is contained in $\{x : x \geq \mathbf{0}\}$, the value computed for the heuristic filtering is greater than or equal to the lower bound of Proposition~\ref{pricing:exact:prop}. Hence, we also evaluate the filtering heuristics for both problems. 

In our experiments, we keep in memory dual values obtained at each iteration. However, we compare three strategies for filtering a subproblem associated with some $k \in K$:
\begin{itemize}
    \item {\small \textsc{all}}: all values stored in $\Upsilon_k$ are used to compute lower bounds, as described in Algorithm~\ref{alg:dwd_pricing_filtering},
    \item {\small \textsc{computed}}: only one lower bound is computed using the last value stored in  $\Upsilon_k$,
    \item {\small \textsc{add}}: only one lower bound is computed using the last value $(\ell, \bar{c}^{k,\ell},\mu_\ell^k)$ with $\bar{c}^{k,\ell} <0$ stored in  $\Upsilon_k$.
\end{itemize}

The first strategy consists in computing all possible bounds whereas the second (resp. third) strategy uses only the information corresponding to the last time the pricing subproblem associated with $k \in K$ was solved (resp. a negative cost column associated with $k$ was found). 
All strategies have been tried for exact and heuristic filterings. In the experimental results, the name of the strategy is preceded by {\small \textsc{exact-}} in case of exact filtering, and {\small \textsc{heur-}} otherwise.

\medskip 

Each strategy is compared with the baseline algorithm (referred hereafter as {\small \textsc{baseline}}) that corresponds to DWD without any filtering. For the baseline algorithm, the number of solved pricing subproblems (\texttt{\#Calls}), the number of columns in the last RMP (\texttt{\#Vars}) and the running time in seconds (\texttt{time (s)}) of DWD is reported. For each strategy, we report in the tables:
\begin{itemize}
\setlength\itemsep{-0.1em}
\item[--] \texttt{\%rCalls} is the percentage reduction of the number of solved pricing subproblems with respect to the baseline. 
It is equal to $100 \times \frac{\texttt{\#Calls} - n}{\texttt{\#Calls}}$ for a filtering strategy that solves $n$ pricing subproblems in a DWD for an instance whereas baseline solves \texttt{\#Calls} ones.
\item[--] \texttt{\%rTime} is the percentage reduction of the running time with respect to the baseline. It is equal to $100 \times \frac{\texttt{time (s)} - t}{\texttt{time (s)}}$ for a filtering strategy that solves DWD in $t$ seconds for an instance whereas baseline solves it in \texttt{time (s)} seconds.

\end{itemize}
For heuristic filtering methods, we also provide the gap (column \texttt{GAP}) to measure the quality of the returned solution. It is equal to $100\times \frac{v - v^*}{v^*}$ where $v$ is the cost of the returned solution and $v^*$ the cost of the solution obtained with no filtering method.

\medskip

Experiments were conducted on an Intel(R) Xeon(R) CPU E5-4627 v2 of 3.30GHz with 504GB RAM, running under Linux 64 bits. A vanilla CG scheme was developed in C++ without any framework, and CPLEX 12.6.3 was used to solve RMP at each iteration of CG and the pricing subproblems for the generalized assignment problem. A precision (\textit{e.g.}, pricing tolerance) of $10^{-4}$ is used when checking whether a column has a negative reduced cost to escape the numerical inaccuracies complexations~\citep{colgen}.

\medskip

The rest of this section is divided into two parts: the first part is dedicated to the multi-commodity flow problem with side constraints whereas the second one is related to the generalized assignment problem. In each part, we describe the problem and DWD for solving the linear relaxation and the experimental results we obtain. 

\subsection{Multi-Commodity Flow Problem with Side Constraints}

The \emph{Multi-Commodity Flow Problem with Side Constraints (MC)} consists in routing a set of commodities in a capacitated network from their source to their target while minimizing the cost and respecting the capacity and the delay constraints \cite{Holmberg}.

An instance of MC is defined as follows. Given a directed graph $G = (V, A)$ where $V$ is the set of vertices and $A$ is the set of arcs, we associate with each arc $a \in A$  a capacity $c_a$, a delay $d_a$ ($\geq 0$) and a routing cost $r_a$ ($\geq 0$). For each vertex $v \in V$, we define $\delta^+(v)$ as the set of outgoing arcs and $\delta^-(v)$ as the set of ingoing arcs. Furthermore, given a set of commodities $K$, for each commodity $k$, we associate a source vertex $s_k \in V$, a target vertex $t_k \in V$, a demand bandwidth $b_k$, a maximal delay $d_k$.

\paragraph{DWD} When not considering delay requirements, the multiflow problem can be formulated with a compact integer linear formulation where variables are the quantity of commodity $k \in K$ that is routed on arc $a \in A$. This formulation is known in the literature as the arc-flow formulation \cite{Ahuja1993}. However, even if this formulation is compact, solving its linear relaxation tends to be challenging for large instances due to the number of variables and constraints. It is preferable to apply DWD to solve it. This gives the well-known arc-path linear relaxation. Taking into account delay requirements can be done by integrating it in the pricing, that is, by considering only paths satisfying the delay requirement. More formally, for each commodity $k \in K$, let $P^k$ be the set of $s_kt_k$-paths $p$ of $G$ respecting the delay requirement $\sum_{a \in p} d_a \le d_k$. Let $\mathcal{X}^k = \{\chi^p:p \in P^k\}$ be the set of incidence vectors of paths of $P^k$. 

Let $\lambda_p^k$ denote the fraction of commodity $k$ that is routed along path $p$, and let $r_p = \sum_{a \in p} r_a$ denote the routing cost of path $p$. The linear relaxation of MC is as follows~\cite{Holmberg}:  
\begin{align}
&\min &&\sum_{k \in K} \sum_{p \in P^k} b_k \sum_{a \in p} r_a \, \lambda_p^k \label{cmcfp:obj}\\
&\text{s.t. }&&-\sum_{k \in K} b_k \sum_{p\in P^k \mid \, a \in p}  \lambda_p^k \ge -c_a  & \forall a \in A \label{cmcfp:capacity}\\
& &&\sum_{p\in P^k} \lambda_p^k \ge 1  &  \forall k \in K \label{cmcfp:convexity} \\
& &&\lambda_p^k \ge 0  & \forall k \in K, \forall p \in P^k \label{cmcfp:nonnegativity}
\end{align}

The objective function~\eqref{cmcfp:obj} minimizes the total routing cost. Inequalities \eqref{cmcfp:capacity} are the capacity constraints that ensure that the total flow over an arc does not exceed its capacity. The convexity constraints \eqref{cmcfp:convexity} impose that each demand is routed. Note that in DWD formulation, constraints~\eqref{cmcfp:convexity} are equalities, but this can be relaxed without changing the optimality of the linear relaxation. Constraints \eqref{cmcfp:nonnegativity} are the nonnegativity constraints.

Formulation~\eqref{cmcfp:obj}-\eqref{cmcfp:nonnegativity} contains an exponential number of nonnegative variables. The columns are then dynamically generated, and the pricing problem is decomposed into a subproblem for each commodity. The pricing subproblem associated with the commodity $k \in K$ consists in determining a path $p$ of $P^k$ that minimizes the reduced cost $\sum_{a \in p} (r_a + b_k \pi_a) - \mu_k$. This reduces to computing a resource-constrained shortest $s_kt_k$-path problem with respect to arc costs $w_a = r_a + b_k \pi_a$~\cite{CSPP}. In our experiments, this problem is solved using the algorithm described in~\cite{LOZANO2013378}.

\paragraph{Filters}  

Using the binary hypercube for $Y^k$, at iteration $t$, the lower bound corresponding to the exact filtering on $\bar{c}^{k,t}$ using the information of iteration $\ell$  is given by:
\begin{equation}\label{eq:exact_filter_mc}
    \bar{c}^{k,\ell} + \mu_k^\ell - \mu_k^t + \sum_{a \in A} \min\{0, (\pi_a^\ell - \pi_a^t)\}
\end{equation}
Such a bound can be computed in $O(|A|)$.
The heuristic filter consists in restricting the sum in equation~\eqref{eq:exact_filter_mc} to the set of arcs contained in a path of $P^k$ at iteration~$t$.

\paragraph{Dataset}

In our experiments, we used SNDlib instances~\cite{OrlowskiPioroTomaszewskiWessaely2010}, and realistic telecommunication instances detailed in \cite{HUIN202372}.
In the tables reporting the results we obtained on these instances, the number $|V|$ of nodes, $|A|$ of arcs and $|K|$ of commodities  is given in parentheses after the name of the instance.

\paragraph{Experimental results}

\medskip

An experimental comparison of the different filtering strategies for MC is presented in Tables \ref{experiments:table:MCF} and~\ref{experiments:table:MCF:heuristic}. 

Table \ref{experiments:table:MCF} shows the results for exact filtering methods. 
First, let us remark that for SNDlib instances, the computational time is small, \textit{i.e.}, less than two seconds for each instance. Consequently, the gain in terms of time can fluctuate between several runs. Moreover, for these instances, the time spent on the pricing subproblems is small. Therefore, the extra burden due to the copies of the dual variables in memory and the computation of the lower bounds may not compensate for the time saved thanks to the filtered pricing subproblems. However, on Nn instances, as the computational time is longer, filtering speeds up a lot DWD. For a few instances, it allows us to filter around half of the pricing subproblems, dividing more or less by two the running time of DWD. 
We remark that the best gain in computational time is not always \textsc{exact-all}. Indeed, \textsc{exact-add} provides better computational time for some Nn instances. This is because, even if \textsc{exact-add} filters less subproblems, it needs less time at each iteration to check whether or not a pricing subproblem is filtered, as it computes only one lower bound per subproblem. Finally, on some instances (dfn-bwin, india35, newyork, norway and zib54), all pricing subproblems with no negative reduced cost columns are filtered by every strategy.

Table \ref{experiments:table:MCF:heuristic} shows the performance when heuristic pricing filtering methods are used.
The \textsc{heur-all} method provides the best performance in terms of the number of filtered pricing subproblems. For instance, \textsc{heur-all} filters 21\%  more subproblems than \textsc{exact-all}.
 
For SNDlib instances where the computational time is low, the computational time is reduced on average by 15\%. For Nn instances, where the computational time is higher, the best method is \textsc{heur-computed} where the computational time is reduced on average by 60\%. Even if \textsc{heur-computed} slightly filters less pricing subproblems than \textsc{heur-all}, it has a better computational time on average. 
Strategies \textsc{heur-all} and \textsc{heur-computed} decrease the quality of the linear relaxation by 0.28\% on average, while \textsc{heur-add} always obtains the best solution in the tested instances. Remark that \textsc{heur-add} gets a better performance in comparison to the exact filtering methods. This method is a good trade-off between improvement and quality of the linear relaxation.

To sum up, for MC, the proposed methods provide an important speed-up. The exact filtering methods gain an average of 21\% of time.
Heuristic filtering methods also show good performances: the best strategy decreases the computational time up to 70\% and about 29\% on average, and it provides near-optimal solutions with an average gap of 1.2\%.

\subsection{Generalized assignment problem}

The \emph{Generalized Assignment Problem (GA)} consists in assigning a set of items to a set of bins while minimizing the assignment cost and without exceeding the bin capacities~\cite{gapDWD}.
An instance of GA is as follows: given $m$ items and a set $K = \{ 1 \dots |K|\}$ of bins, we associate with each bin $k$, a capacity $W_k$ ($\geq 0$), and for each item $i$ contained in bin $k$ a cost $c_{i,k}$ ($\geq 0$) and a weight $w_{i,k}$ ($\geq 0$). 

\paragraph{DWD} There exists a compact formulation for GA but its linear relaxation is usually tackled using DWD, resulting in branch-and-price algorithms~\cite{Savelsbergh97branch}.

For each bin $k \in K$, let $P^k$ be the set of feasible assignments for bin $k$, that is, the set of item sets whose weight sum does not exceed the bin capacity. Denote by $\mathcal{X}^k$ the set of incidence vectors of feasible assignments of bin $k$. By definition,  $\mathcal{X}^k = \{x \in \{0,1\}^m : \sum_{i=1}^m w_{i,k} x_i \le W_k\}$. The linear relaxation of GA can be formulated as follows:
\begin{align}
\label{form:GA:obj}&\min &&\sum_{k \in K} \sum_{p \in P^k} (\sum_{i \in p} c_{i,k}) \lambda_p^k & \\
\label{form:GA:assignOneItem} &\text{s.t. }&&\sum_{k \in K} \sum_{p \in P^k | i \in p} \lambda_p^k  \geq 1 &\forall i=1,\dots,m \\
\label{form:GA:atMostOneAssignment} & &&-\sum_{p \in P^k} \lambda_p^k \ge -1 &\forall k \in K \\
\label{form:GA:nonnegativity} & && \lambda_p^k \geq 0 &\forall k \in K, \forall p \in P^k
\end{align}

The objective function \eqref{form:GA:obj} minimizes the assignment cost. Inequalities~\eqref{form:GA:assignOneItem} ensure that each item belongs to at least one selected feasible assignment of a bin. Inequalities~\eqref{form:GA:atMostOneAssignment} and~\eqref{form:GA:nonnegativity} force a bin to be used at most once. Note that GA is usually presented as a maximization problem but it can be easily translated in the minimization form we present, see~\cite{martello_generalized_1992} for details.

Formulation \eqref{form:GA:obj}-\eqref{form:GA:nonnegativity} contains an exponential number of variables, so columns are dynamically generated. The pricing problem can be decomposed into one pricing subproblem for each bin $k \in K$. Denoting by $\pi$ and $\mu$ the dual variables associated with inequalities~\eqref{form:GA:assignOneItem} and~\eqref{form:GA:atMostOneAssignment} respectively, the pricing subproblem associated with bin $k \in K$ reduces to find whether there exists a feasible assignment $p \subseteq \{1,\dots,m\}$ for bin $k$  such that $\sum_{i \in p} (c_{i,k} - \pi_i) + \mu_k$ is negative. This reduces to the computation of a binary knapsack problem with $m$ items, each having a cost equal to $(c_{i,k} - \pi_i)$ and a weight $w_{i,k}$, the knapsack having a capacity of $W_k$. In our experiments, we formulate it as a binary linear problem and solve it using CPLEX.

\paragraph{Filters}

Using the binary hypercube for $Y^k$, at iteration $t$, the lower bound corresponding to the exact filtering on $\bar{c}^{k,t}$ using the information of iteration $\ell$  is given by:
\begin{equation}\label{eq:exact_filter_ga}
    \bar{c}^{k,\ell} - \mu_k^\ell + \mu_k^t + \sum_{i = 1}^m \min\{0, (\pi_i^t - \pi_i^\ell)\}
\end{equation}
Such a bound can be computed in $O(m)$.
The heuristic filter consists in restricting the sum in equation~\eqref{eq:exact_filter_ga} to the set of items contained in a feasible assignment of $P^k$ at iteration~$t$.

\paragraph{Dataset}
Pricing filtering in DWD is interesting when the number of solved pricing subproblems with no negative reduced cost column is large. For GA, when the number of objects is bigger than the number of bins, this number is small. Hence, our method is ineffective for GA instances of the OR-Library \cite{orlib} so we generate 9 sets of random instances, denoted E1$,\dots, $ E9,  where the number of bins is greater than the number of objects.
Each set corresponds to 100 randomly generated instances containing the same number of bins (100, 1000 or 5000) and objects (10, 50, or 100). In the tables reporting the results we obtained for these sets, the number $|K|$ of bins and $m$ of objects are given in parentheses after the name of the set.

For a given number of bins and objects, each of the 100 instances is generated as follows.
The costs and weights are randomly chosen in the intervals $\llbracket 1,100 \rrbracket$ and $\llbracket 5, 20 \rrbracket$, respectively. To ensure the existence of a solution, bin capacities are defined as follows. An assignement of the objects to the bins is randomly chosen, and the capacity of each bin is set to the sum of the weights of its assigned objects plus one, if there are any, and to a random value within the interval $\llbracket 5,100 \rrbracket$ otherwise.

\paragraph{Experimental results}

Tables \ref{experiments:table:GA} and \ref{experiments:table:GA:heuristic} report the average results obtained for GA with exact and heuristic filtering, respectively. 

We remark that the percentage of filtered pricing subproblems is smaller in comparison with MC. Whereas \textsc{exact-add} does not filter any pricing subproblem in any instance, the other two exact strategies filter some pricing subproblems for five sets of instances over the nine tested sets, and the percentage of filtered subproblems by both strategies are similar. The increase in the computational time for the instances without a filtered subproblem is small for the \textsc{exact-computed} method, and a little bit more for the \textsc{exact-all} method. On average, the \textsc{exact-all} (resp. \textsc{exact-computed}) method filters 6.413\% (resp. 6.406\%) of the pricing subproblems and decreases by 14\% (Resp. 11.16\%) the computational time. 

Regarding the heuristic filtering strategies, the best one is \textsc{heur-computed} which filters 14.23\% of the pricing subproblems, resulting in a decrease of 17.20\% of the computational time, but always finds an optimal solution in every tested instance. The \textsc{heur-all} method provides a better improvement from a computational time point of view with a decrease of 26.76\% on average, but results in non-optimal solutions with a gap up to 23.89\% for an instance set and about 3. 75\% on average. 
Note that for dataset E6, \textsc{heur-add} increases the number of solved pricing subproblems because the number of iterations in DWD increases with respect to the baseline.

\begin{table*}[!ht]
    \centering
    \scalebox{0.9}{
\begin{tabular}{|l|r|r|r|r|r|r|r|r|r|} \hline
\multirow{2}{*}{Instance ($|V|$, $|A|$, $|K|$)} &  \multicolumn{3}{|c|}{{\small \textsc{baseline}}} &  \multicolumn{2}{|c|}{{\small \textsc{exact-add}}} &  \multicolumn{2}{|c|}{{\small \textsc{exact-all}}} &  \multicolumn{2}{|c|}{{\small \textsc{exact-computed}}}  \\ \cline{2-10}
 & \#Calls & \#Added & time (s) & \texttt{\%rCalls}  & \texttt{\%rTime} & \texttt{\%rCalls} & \texttt{\%rTime} & \texttt{\%rCalls} & \texttt{\%rTime}\\
\hline
        dfn-bwin (10, 45, 90)& 180 & 90 & 0.024 & \textbf{50.00\%} & \textbf{16.67\%} & \textbf{50.00\%} & \textbf{16.67\%} & \textbf{50.00\%} & 8.33\% \\ \hline
        nobel-us (14, 21, 91) & 273 & 100 & 0.029 & 0.00\% & -24.14\% & \textbf{1.47\%} & \textbf{-10.34\%} & \textbf{1.47\%} & \textbf{-10.34\%} \\ \hline
        polska (12, 18, 66) & 198 & 70 & 0.031 & 0.00\% & \textbf{19.35\%} & \textbf{4.04\%} & -6.45\% & \textbf{4.04\%} & 9.68\% \\ \hline
        pdh (11, 34, 24) & 120 & 47 & 0.038 & 0.00\% & \textbf{21.05\%} & \textbf{10.83\%} & -10.53\% & \textbf{10.83\%} & 0.00\% \\ \hline
        newyork (16, 46, 240) & 480 & 240 & 0.042 & \textbf{50.00\%} & 14.29\% & \textbf{50.00\%} &\textbf{ 23.81\% }& \textbf{50.00\%} & 9.52\% \\ \hline
        sun (27, 102, 67) & 335 & 153 & 0.053 & 15.52\% & \textbf{9.43\%} & \textbf{16.72\% }& 3.77\% & 14.63\% & 7.55\% \\ \hline
        di-yuan (11, 42, 22) & 176 & 62 & 0.065 & 7.39\% & \textbf{9.23\%} & \textbf{7.95\%} & -3.08\% & \textbf{7.95\%} & 7.69\% \\ \hline
        nobel-germany (17, 26, 121) & 847 & 206 & 0.082 & 6.85\% & 9.76\% & \textbf{16.17\%} & \textbf{26.83\%} & 14.88\% & 24.39\% \\ \hline
        norway (27, 102, 67) & 1404 & 702 & 0.095 & \textbf{50.00\%} & \textbf{35.79\%} & \textbf{50.00\%} & 28.42\% & \textbf{50.00\%} & 30.53\% \\ \hline
        ta1 (24, 55, 396) & 1584 & 453 & 0.099 & 61.68\% & 31.31\% & \textbf{64.33\%} & \textbf{41.41\%} & \textbf{64.33\% }& 28.28\% \\ \hline
        india35 (35, 80, 595) & 1190 & 595 & 0.110 &\textbf{ 50.00\% }& \textbf{39.09\%} & \textbf{50.00\%} & 32.73\% & \textbf{50.00\%} & 36.36\% \\ \hline
        france(25, 45, 300) & 2100 & 661 & 0.184 & 8.95\% & 3.26\% & \textbf{13.05\%} & \textbf{13.04\%} & 12.33\% & 9.78\% \\ \hline
        geant (22, 36, 462) & 3234 & 810 & 0.192 & 18.21\% & 2.60\% & \textbf{31.48\%} & \textbf{4.17\%} & 29.72\% & 0.52\% \\ \hline
        zib54 (54, 81, 1501) & 3002 & 1501 & 0.241 & \textbf{50.00\% }& 43.57\% & \textbf{50.00\%} & \textbf{50.21\%} & \textbf{50.00\%} & 36.93\% \\ \hline
        nobel-eu (28, 41, 378) & 2646 & 683 & 0.272 & 1.85\% & 4.78\% & \textbf{16.14\% }& \textbf{14.71\%} & 14.85\% & -0.37\% \\ \hline
        ta2 (65, 108, 1869) & 11214 & 3331 & 1.098 & 36.99\% & 23.95\% & \textbf{46.70\% }& 29.69\% & \textbf{46.70\%} & \textbf{29.78\%} \\ \hline
        cost266 (37, 57, 1332) & 9324 & 2555 & 1.613 & 3.43\% & -0.37\% & \textbf{11.65\%} & 3.10\% & 11.07\% &\textbf{ 4.03\%} \\ \hline
        N600 (600, 2400, 1000) & 7000 & 1546 & 5.661 & 42.40\% &\textbf{ 40.66\%} & \textbf{45.00\%} & 40.43\% & 40.93\% & 39.30\% \\ \hline
        N800 (800, 3200, 2000) & 14000 & 3438 & 15.036 & 35.09\% & \textbf{32.73\%} & \textbf{37.99\%} & 30.91\% & 34.00\% & 31.78\% \\ \hline
        N1000 (1000, 4000, 4000) & 36000 & 7442 & 48.988 & 28.78\% &\textbf{ 26.57\%} & \textbf{31.35\%} & 22.00\% & 27.67\% & 25.52\% \\ \hline
        N1400 (1400, 5600, 10000) & 60000 & 13552 & 111.846 & 53.19\% & 51.14\% & \textbf{55.32\% }& \textbf{51.47\%} & 52.15\% & 50.23\% \\ \hline
        N1200 (1200, 4800, 8000) & 72000 & 15783 & 121.757 & 24.81\% & \textbf{21.21\%} & \textbf{27.43\%} & 16.71\% & 23.86\% & 20.56\% \\ \hline
        N1600 (1600, 6400, 12000) & 72000 & 16189 & 155.376 & 53.08\% & 51.32\% & \textbf{55.03\%} & \textbf{51.58\%} & 52.03\% & 50.36\% \\ \hline
        N1800 (1800, 7200, 15999) & 127992 & 26772 & 327.129 & 35.34\% & 30.41\% & \textbf{37.51\%} & \textbf{31.91\%} & 33.71\% & 28.90\% \\ \hline
        N2000 (2000, 8000, 16000) & 144000 & 25757 & 407.631 & 40.00\% & 35.83\% & \textbf{42.87\%} & \textbf{37.28\% }& 38.36\% & 34.56\% \\ \hline
    \end{tabular}
}
    \caption{Experimental results on MC with different exact filtering strategies}
    \label{experiments:table:MCF}
\end{table*}

\begin{table*}[!ht]
\centering
\scalebox{0.75}{%
\begin{tabular}{|l|r|r|r|r|r|r|r|r|r|r|r|r|} \hline

\multirow{2}{*}{Instance ($|V|$, $|A|$, $|K|$)} & \multicolumn{3}{|c|}{{\small \textsc{baseline}}} &  \multicolumn{3}{|c|}{{\small \textsc{heur-add}}} &  \multicolumn{3}{|c|}{{\small \textsc{heur-all}}} &  \multicolumn{3}{|c|}{{\small \textsc{heur-computed}}}  \\ \cline{2-13}

& \texttt{\#Calls} & \texttt{time (s)} & \texttt{cost} & \texttt{\%rCalls}  & \texttt{\%rTime} & \texttt{GAP} &\texttt{\%rCalls} & \texttt{\%rTime} & \texttt{GAP} & \texttt{\%rCalls} & \texttt{\%rTime} & \texttt{GAP}\\
\hline
dfn-bwin (10, 45, 90) & 180 & 0.024 & 1.83E+07 & \textbf{50.00\%} & \textbf{16.67\%} & 0.00\% & \textbf{50.00\%} & 8.33\% & 0.00\% & \textbf{50.00\%} & 8.33\% & 0.00\% \\ \hline
        nobel-us (14, 21, 91) & 273 & 0.029 & 8.06E+08 & 0.00\% & \textbf{-6.90\%} & 0.00\% & \textbf{1.10\%} & -10.34\% & 0.00\% & \textbf{1.10\%} & -13.79\% & 0.00\% \\ \hline
        polska (12, 18, 66) & 198 & 0.031 & 7.59E+08 & 0.00\% & \textbf{12.90\%} & 0.00\% & \textbf{3.54\%} & 9.68\% & 0.00\% & \textbf{3.54\%} & \textbf{12.90\%} & 0.00\% \\ \hline
        pdh (11, 34, 24) & 120 & 0.038 & 2.67E+08 & 0.00\% & \textbf{5.26\%} & 0.00\% & 0.00\% & 0.00\% & 0.00\% & 0.00\% & -10.53\% & 0.00\% \\ \hline
        newyork (16, 46, 240) & 480 & 0.042 & 8.46E+04 & \textbf{50.00\%} & 16.67\% & 0.00\% & \textbf{50.00\%} & \textbf{30.95\%} & 0.00\% & \textbf{50.00\%} & 23.81\% & 0.00\% \\ \hline
        sun (27, 102, 67) & 335 & 0.053 & 2.71E+06 & 15.52\% & 3.77\% & 0.00\% & \textbf{28.06\%} & \textbf{18.87\%} & 0.00\% & 27.46\% & 1.89\% & 0.00\% \\ \hline
        di-yuan (11, 42, 22) & 176 & 0.065 & 1.40E+06 & 7.39\% & -4.62\% & 0.00\% & \textbf{19.89\%} & 3.08\% & 0.00\% & \textbf{19.89\%} & \textbf{6.15\%} & 0.00\% \\ \hline
        nobel-germany (17, 26, 121) & 847 & 0.082 & 6.96E+07 & 6.85\% & 13.41\% & 0.00\% & \textbf{25.62\%} & 15.85\% & 0.00\% & 25.50\% & \textbf{18.29\%} & 0.00\% \\ \hline
        norway (27, 102, 67) & 1404 & 0.095 & 7.05E+05 & \textbf{50.00\%} & \textbf{33.68\%} & 0.00\% & \textbf{50.00\%} & 29.47\% & 0.00\% & \textbf{50.00\%} & 32.63\% & 0.00\% \\ \hline
        ta1 (24, 55, 396) & 1584 & 0.099 & 1.39E+11 & 62.12\% & \textbf{43.43\%} & 0.00\% & \textbf{65.47\%} & 39.39\% & 0.00\% & \textbf{65.47\%} & 39.39\% & 0.00\% \\ \hline
        india35 (35, 80, 595) & 1190 & 0.110 & 3.51E+05 & \textbf{50.00\% }& 31.82\% & 0.00\% & \textbf{50.00\% }& 28.18\% & 0.00\% & \textbf{50.00\%} & \textbf{33.64\%} & 0.00\% \\ \hline
        france (25, 45, 300) & 2100 & 0.184 & 1.90E+10 & 9.00\% & 12.50\% & 0.00\% & \textbf{19.57\%} & \textbf{16.30\%} & 0.00\% & 19.33\% & 9.24\% & 0.00\% \\ \hline
        geant (22, 36, 462) & 3234 & 0.192 & 3.54E+11 & 18.68\% & 0.52\% & 0.00\% & \textbf{42.39\% }& 5.21\% & 0.04\% & 41.84\% & \textbf{18.75\%} & 0.04\% \\ \hline
        zib54 (54, 81, 1501) & 3002 & 0.241 & 1.18E+06 & \textbf{50.00\%} & 43.57\% & 0.00\% & \textbf{50.00\%} & 43.15\% & 0.00\% & \textbf{50.00\%} & \textbf{43.57\%} & 0.00\% \\ \hline
        nobel-eu (28, 41, 378) & 2646 & 0.272 & 1.57E+09 & 2.00\% & \textbf{3.31\%} & 0.00\% & \textbf{10.54\%} & -4.78\% & 0.00\% & 9.90\% & -0.37\% & 0.00\% \\ \hline
        ta2 (65, 108, 1869) & 11214 & 1.098 & 3.31E+12 & 37.90\% & 27.14\% & 0.03\% & \textbf{55.83\% }& 38.07\% & 1.74\% & \textbf{55.83\%} & \textbf{39.16\%} & 1.74\% \\ \hline
        cost266 (37, 57, 1332) & 9324 & 1.613 & 1.49E+12 & 3.43\% & -0.68\% & 0.00\% & \textbf{13.04\%} & 0.06\% & 0.00\% & 12.82\% & \textbf{1.12\%} & 0.00\% \\ \hline
        N600 (600, 2400, 1000) & 7000 & 5.661 & 2.08E+07 & 42.51\% & 42.78\% & 0.00\% & \textbf{66.60\%} & 61.44\% & 0.01\% & 65.71\% &\textbf{ 63.12\%} & 0.01\% \\ \hline
        N800 (800, 3200, 2000) & 14000 & 15.036 & 5.91E+10 & 35.18\% & 34.03\% & 0.00\% & \textbf{60.04\%} & 55.17\% & 0.00\% & 58.69\% & \textbf{56.00\%} & 0.00\% \\ \hline
        N1000 (1000, 4000, 4000) & 36000 & 48.988 & 1.06E+12 & 28.88\% & 27.30\% & 0.00\% & \textbf{58.21\%} & 53.32\% & 2.43\% & 56.53\% & \textbf{54.45\%} & 2.43\% \\ \hline
        N1400 (1400, 5600, 10000) & 60000 & 111.846 & 3.66E+11 & 53.27\% & 51.90\% & 0.00\% & \textbf{69.20\% }& 66.37\% & 0.00\% & 68.44\% & \textbf{67.09\%} & 0.00\% \\ \hline
        N1200 (1200, 4800, 8000) & 72000 & 121.757 & 7.05E+12 & 24.89\% & 22.48\% & 0.00\% & \textbf{49.86\% }& 42.85\% & 1.51\% & 47.82\% & \textbf{43.98\%} & 1.51\% \\ \hline
        N1600 (1600, 6400, 12000) & 72000 & 155.376 & 1.26E+12 & 53.16\% & 51.95\% & 0.00\% & \textbf{68.88\%} & 66.49\% & 0.00\% & 68.06\% &\textbf{ 67.00\%} & 0.00\% \\ \hline
        N1800 (1800, 7200, 15999) & 127992 & 327.129 & 3.08E+12 & 35.44\% & 31.42\% & 0.00\% & \textbf{64.15\%} & \textbf{60.92\%} & 1.44\% & 62.98\% & 60.06\% & 1.44\% \\ \hline
        N2000 (2000, 8000, 16000) & 144000 & 407.631 & 6.69E+12 & 40.13\% & 36.90\% & 0.00\% & \textbf{70.69\%} & \textbf{69.16\%} & 0.00\% & 69.70\% & 68.31\% & 0.00\% \\ \hline
\end{tabular}%
}
\caption{Experimental results on MC with different heuristic filtering strategies}
    \label{experiments:table:MCF:heuristic}
\end{table*}

\begin{table*}[!ht]
\centering
\begin{tabular}{|l|r|r|r|r|r|r|r|r|r|} \hline

\multirow{2}{*}{Instance ($|K|$, $m$)} &  \multicolumn{3}{|c|}{{\small \textsc{baseline}}} &  \multicolumn{2}{|c|}{{\small \textsc{exact-add}}} &  \multicolumn{2}{|c|}{{\small \textsc{exact-all}}} &  \multicolumn{2}{|c|}{{\small \textsc{exact-computed}}}  \\ \cline{2-10}

& \texttt{\#Calls} & \texttt{\#Added} & \texttt{time (s)} & \texttt{\%rCalls}  & \texttt{\%rTime} & \texttt{\%rCalls} & \texttt{\%rTime} & \texttt{\%rCalls} & \texttt{\%rTime}\\
\hline
E1 (100, 10) & 800 & 456 & 0.49 & 0.00\% & -6.12\% & \textbf{12.13\%} & \textbf{40.82\%} & \textbf{12.13\%} & 16.33\% \\ \hline
E2 (100, 50) & 1600 & 945 & 1.10 & 0.00\% & 10.00\% & \textbf{6.25\%} &\textbf{ 27.27\%} & 6.19\% & 15.45\% \\ \hline
E3 (100, 100) & 2500 & 1897 & 1.61 & 0.00\% & 11.80\% & \textbf{3.96\%} & 16.15\% & \textbf{3.96\%} & \textbf{19.25\%} \\ \hline
E4 (1000, 10) & 5000 & 2708 & 2.58 & 0.00\% & 20.54\% & \textbf{19.10\%} & \textbf{41.86\%} & \textbf{19.10\%} & 38.37\% \\ \hline
E5 (1000, 50) & 8000 & 4882 & 3.08 & 0.00\% & -2.60\% & 0.00\% & -3.57\% & 0.00\% & \textbf{-0.65\%} \\ \hline
E6 (1000, 100) & 8000 & 6053 & 3.56 & 0.00\% &\textbf{ -0.84\% }& 0.00\% & -3.09\% & 0.00\% & -1.69\% \\ \hline
E7 (5000, 10) & 20000 & 11744 & 6.68 & 0.00\% & 0.30\% & \textbf{16.28\%} & 13.77\% & \textbf{16.28\%} & \textbf{14.67\%} \\ \hline
E8 (5000, 50) & 30000 & 17830 & 11.34 & 0.00\% &\textbf{ -0.53\% }& 0.00\% & -2.12\% & 0.00\% & -0.97\% \\ \hline
E9 (5000, 100) & 35000 & 20274 & 15.21 & 0.00\% & -1.12\% & 0.00\% & -2.24\% & 0.00\% &\textbf{ -0.33\%} \\ \hline
\end{tabular}
\caption{Experimental results on GA with different exact filtering strategies}
\label{experiments:table:GA}
\end{table*}

\begin{table*}[!ht]
\centering
\scalebox{0.8}{%
\begin{tabular}{|l|r|r|r|r|r|r|r|r|r|r|r|r|} \hline

\multirow{2}{*}{Instance ($|K|$, $m$)} &  \multicolumn{3}{|c|}{{\small \textsc{baseline}}} &  \multicolumn{3}{|c|}{{\small \textsc{heur-add}}} &  \multicolumn{3}{|c|}{{\small \textsc{heur-all}}} &  \multicolumn{3}{|c|}{{\small \textsc{heur-computed}}}  \\ \cline{2-13}

& \texttt{\#Calls} & \texttt{time (s)} & \texttt{cost} & \texttt{\%rCalls}  & \texttt{\%rTime} & \texttt{GAP} &\texttt{\%rCalls} & \texttt{\%rTime} & \texttt{GAP} & \texttt{\%rCalls} & \texttt{\%rTime} & \texttt{GAP}\\
\hline
E1 (100, 10) & 800 & 0.49 & 17.59 & 0.13\% & -6.12\% & 0.00\% & \textbf{31.38\%} & \textbf{30.61\%} & 3.81\% & 17.25\% & 8.16\% & 0.00\% \\ \hline
E2 (100, 50) & 1600 & 1.10 & 78.57 & 0.00\% & 11.82\% & 0.00\% & \textbf{35.44\%} & \textbf{43.64\%} & 6.13\% & 16.44\% & 24.55\% & 0.00\% \\ \hline
E3 (100, 100) & 2500 & 1.61 & 198.4 & 0.00\% & 10.56\% & 0.00\% & \textbf{29.72\% }& \textbf{33.54\%} & 23.89\% & 5.96\% & 23.60\% & 0.00\% \\ \hline
E4 (1000, 10) & 5000 & 2.58 & 10 & 0.08\% & 21.71\% & 0.00\% & \textbf{26.18\%} & \textbf{47.29\%} & 0.00\% & 22.20\% & 41.09\% & 0.00\% \\ \hline
E5 (1000, 50) & 8000 & 3.08 & 50 & 0.05\% & -0.32\% & 0.00\% & \textbf{27.13\%} & \textbf{20.78\% }& 0.00\% & 15.64\% & 13.64\% & 0.00\% \\ \hline
E6 (1000, 100) & 8000 & 3.56 & 100 & -12.38\% & -13.76\% & 0.00\% & \textbf{12.06\%} & \textbf{7.02\%} & 0.00\% & 0.10\% & -0.84\% & 0.00\% \\ \hline
E7 (5000, 10) & 20000 & 6.68 & 10 & 0.00\% & -1.65\% & 0.00\% & \textbf{16.28\%} & 13.47\% & 0.00\% & \textbf{16.28\%} & \textbf{14.22\%} & 0.00\% \\ \hline
E8 (5000, 50) & 30000 & 11.34 & 50 & 0.03\% & -0.79\% & 0.00\% & \textbf{24.10\%} & \textbf{20.46\%} & 0.00\% & 14.14\% & 12.17\% & 0.00\% \\ \hline
E9 (5000, 100) & 35000 & 15.21 & 100 & 0.03\% & -0.92\% & 0.00\% & \textbf{28.40\% }& \textbf{24.06\% }& 0.00\% & 20.09\% & 17.29\% & 0.00\% \\ \hline
\end{tabular}%
}

\caption{Experimental results on GA with different heuristic filtering strategies \label{experiments:table:GA:heuristic}}
\end{table*}

\section{Concluding remarks}

From our knowledge, this is the first work that proposes a filtering method based on the information gathered in previous iterations.
Since the impact of the filtering depends on the problem and the decomposition, more computational experiments and research should be done in order to determine for which kind of problems such a filtering method is interesting.
Moreover, heuristic and exact methods could be combined, that is, a heuristic filtering may be used in early iterations of DWD, and may be switched to exact filtering when no negative reduced cost column is found. In our experiments, we did not manage to speed-up DWD by doing this since in MC, the number of iterations of DWD is very small, and in GA, the gain with the heuristic filtering is not enough to overcome the extra iteration of DWD with exact filtering.
Finally, it should be interesting to develop these filtering methods inside a generic column generation framework to promote its adoption.

\bibliographystyle{elsarticle-num}
\bibliography{bibliography.bib}

\end{document}